\documentclass[reqno]{amsart}%
\usepackage{amsfonts}
\usepackage{verbatim}
\usepackage{xcolor}
\usepackage{amsmath}
\usepackage{amssymb}
\usepackage{graphicx}
\usepackage{accents}%
\providecommand{\U}[1]{\protect\rule{.1in}{.1in}}
\newtheorem{theorem}{Theorem}[section]

\newtheorem{proposition}[theorem]{Proposition}

\numberwithin{equation}{section}
\begin{document}
\title[KdV equation]{On classical solutions of the KdV equation}
\author{Sergei Grudsky}
\address{Departamento de Matematicas, CINVESTAV del I.P.N. Aportado Postal 14-740,
07000 Mexico, D.F., Mexico.}
\email{grudsky@math.cinvestav.mx.}
\author{Alexei Rybkin}
\address{Department of Mathematics and Statistics, University of Alaska Fairbanks, PO
Box 756660, Fairbanks, AK 99775}
\email{arybkin@alaska.edu}
\thanks{SG is supported by CONACYT grant 238630. AR is supported in part by the NSF
grant DMS-1716975.}
\date{May, 2019}
\subjclass{34L25, 37K15, 47B35}
\keywords{KdV equation, Hankel operators.}

\begin{abstract}
We show that if the initial profile $q\left(  x\right)  $ for the Korteweg-de
Vries (KdV) equation is essentially semibounded from below and $\int^{\infty
}x^{5/2}\left\vert q\left(  x\right)  \right\vert dx<\infty,$ (no decay at
$-\infty$ is required) then the KdV has a unique global classical solution
given by a determinant formula. This result is best known to date.

\end{abstract}
\maketitle
\dedicatory{We dedicate this paper to the memory of Jean Bourgain.}

\section{Introduction}

We are concerned with the Cauchy problem for the Korteweg-de Vries (KdV)
equation%
\begin{equation}%
\begin{cases}
\partial_{t}u-6u\partial_{x}u+\partial_{x}^{3}u=0,\ \ \ x\in\mathbb{R}%
,t\geq0\\
u(x,0)=q(x).
\end{cases}
\label{KdV}%
\end{equation}
As is well-known, (\ref{KdV}) is the first nonlinear evolution PDE solved in
the seminal 1967 Gardner-Greene-Kruskal-Miura paper \cite{GGKM67} by the
method which is now referred to as the inverse scattering transform (IST).
Much of the original work was done under generous assumptions on initial data
$q$ (typically from the Schwartz class) for which the well-posedness of
(\ref{KdV}) was not an issue even in the classical sense\footnote{I.e., at
least three times continuously differentiable in $x$ and once in $t$.}. But
well-posedness in less nice function classes becomes a problem. The main (but
of course not the only) difficulty is related to slower decay of $q$ at
infinity which negatively affects regularity of the solutions. This issue drew
much of attention once (\ref{KdV}) became in the spot light. For the earlier
literature account we refer the reader to the substantial 1987 paper
\cite{CohenKappSIAM87} by Cohen-Kappeler. The main result of
\cite{CohenKappSIAM87} says that if\footnote{$\int^{\infty}\left\vert f\left(
x\right)  \right\vert dx<\infty$ means that $\int_{a}^{\infty}\left\vert
f\left(  x\right)  \right\vert dx<\infty$ for all finite $a$.}%
\begin{align}
\int_{-\infty}^{\infty}\left(  1+\left\vert x\right\vert \right)  \left\vert
q\left(  x\right)  \right\vert dx &  <\infty,\ \ \label{Cohen Kappeler cond 1}%
\\
\int^{\infty}\left(  1+\left\vert x\right\vert \right)  ^{N}\left\vert
q\left(  x\right)  \right\vert dx &  <\infty,\ \ N\geq
11/4\label{Cohen Kappeler cond 2}%
\end{align}
then (\ref{KdV}) has a classical solution, the initial condition being
satisfied in the Sobolev space $H^{-1}\left(  a,\infty\right)  $ for any real
$a$. The uniqueness was not proven in \cite{CohenKappSIAM87} and in fact it
was stated as an open problem. The best known uniqueness result back then was
available for $H^{3/2}\left(  \mathbb{R}\right)  $ which of course assumes
some smoothness whereas the conditions (\ref{Cohen Kappeler cond 1}%
)-(\ref{Cohen Kappeler cond 2}) do not. Since any function subject to
(\ref{Cohen Kappeler cond 1})-(\ref{Cohen Kappeler cond 2}) can be properly
included in $H^{s}\left(  \mathbb{R}\right)  $ with any negative $s$, a
well-posedness statement in $H^{s}\left(  \mathbb{R}\right)  ,s<0$, would turn
the Cohen-Kappeler existence result into a classical well-posedness. The $s=0$
bar was reached in 1993 in the seminal papers by Bourgain \cite{Bourgain1993}
where, among others, he proved that (\ref{KdV}) is well-posed in $L^{2}\left(
\mathbb{R}\right)  $. Moreover his trademark harmonic analysis techniques
could be pushed below $s=0$. We refer the interested reader to the influential
\cite{ColKeStaTao03} for the extensive literature prior to 2003. Until very
recently, the best well-posedness Sobolev space for (\ref{KdV}) remained
\cite{Killip2018} $H^{-3/4}\left(  \mathbb{R}\right)  $. Note that harmonic
analysis methods break down while crossing $s=-3/4$ in an irreparable way.
Further improvements required utilizing complete integrability of the KdV. The
breakthrough has just occurred in Killip-Visan \cite{Killip2018} where $s=-1$
was reached. That is, (\ref{KdV}) is well-posed for initial data of the form
$q=v+w^{\prime}$ where $v,w\in L^{2}\left(  \mathbb{R}\right)  $. For $s<-1$
the KdV is ill-posed in $H^{s}\left(  \mathbb{R}\right)  $ scale (see
\cite{Killip2018} for relevant discussions and the literature cited therein).

However all these spectacular achievements do not answer the natural question
about the optimal rate of decay of initial data guaranteeing the existence of
a classical solution to (\ref{KdV}) free of a priori smoothness of $q$?
Surprisingly enough, this important question seems to have been in the shadow
and to the best of our knowledge the Cohen-Kappeler conditions
(\ref{Cohen Kappeler cond 1})-(\ref{Cohen Kappeler cond 2}) have not been
fully improved. The current paper is devoted to this question. In particular,
we prove

\begin{theorem}
[Main Theorem]\label{MainThm}Suppose that a real locally integrable initial
profile $q$ in (\ref{KdV}) satisfies:%
\begin{equation}
\operatorname*{Sup}\limits_{\left\vert I\right\vert =1}\int_{I}\max\left(
-q\left(  x\right)  ,0\right)  \ dx<\infty,\ \ \ \ \ \text{(essential
boundedness from below)};\label{cond 1}%
\end{equation}%
\begin{equation}
\int^{\infty}\left(  1+\left\vert x\right\vert \right)  ^{N}\left\vert
q\left(  x\right)  \right\vert dx<\infty,\ \ \ N\geq5/2\ \ \ (\text{rate of
decay at}+\infty),\label{cond 2}%
\end{equation}
then the KdV equation has a unique classical solution $u\left(  x,t\right)  $
such that uniformly on compacts in $\mathbb{R\times R}_{+}$%
\begin{equation}
u\left(  x,t\right)  =\lim_{b\rightarrow-\infty}u_{b}(x,t),\label{convergence}%
\end{equation}
where $u_{b}(x,t)$ is the classical solution with the data $q_{b}=\left.
q\right\vert _{\left(  b,\infty\right)  }$. 
\end{theorem}

We now discuss how Theorem \ref{MainThm} is related to previously known
results and outline the ideas behind our arguments.

Compare first conditions (\ref{Cohen Kappeler cond 1}) and (\ref{cond 1}).
Note that (\ref{Cohen Kappeler cond 1}) is the natural condition for
solubility of the classical inverse scattering problem (the Marchenko
characterization of scattering data \cite{MarchBook2011}), which is the
backbone of the IST. Since the Cohen-Kappeler approach is based upon the
Marchenko integral equation, the condition (\ref{Cohen Kappeler cond 1})
cannot be relaxed within their framework. It is well-known however that the
KdV equation is strongly unidirectional (solitons run to the right) which has
to be reflected somehow in the conditions on initial data. As opposed to
Cohen-Kappeler our approach is based on "one-sided" scattering (from the
right) for the full line Schr\"{o}dinger operator $\mathbb{L}_{q}%
=-\partial_{x}^{2}+q\left(  x\right)  $, which requires the decay\footnote{In
fact only $L^{1}$ decay is needed for the direct scattering problem.}
(\ref{Cohen Kappeler cond 1}) only at $+\infty$. The direct scattering problem
can be solved then as long as $q$ is in the so-called limit point case at
$-\infty$, which is readily provided by our (\ref{cond 1}). But of course the
IST requires by definition a suitable inverse scattering. We however do not
analyze the inverse scattering problem which could in fact be a difficult
endeavour. Instead, we bypass it by considering first truncated data
$q_{b}=\left.  q\right\vert _{\left(  b,\infty\right)  }$ covered by the
classical Faddeev-Marchenko inverse scattering theory. Since $q_{b}\in
H^{-1}\left(  \mathbb{R}\right)  $ for any $b$, the problem (\ref{KdV}) is
well-posed in $H^{-1}\left(  \mathbb{R}\right)  $ (in fact in $H^{s}\left(
\mathbb{R}\right)  $ for any $s<0$). We then study its solution $u_{b}(x,t)$
as $b\rightarrow-\infty$ and it is how our notion of well-posedness comes
about in Theorem \ref{MainThm}. Justifications of our limiting procedures rely
on some subtle facts from the theory of Hankel operators. As the reader will
see in Sections \ref{HO}-\ref{main} the Hankel operator plays an indispensable
role in proving our results. We only mention here that our Hankel operator is
nothing but a different representation of the classical Marchenko operator.
But of course it makes all the difference. Observe that condition
(\ref{cond 1}) doesn't assume any pattern of behavior at $-\infty$ and is, in
a certain sense, optimal (see Section \ref{conclusions}). We noticed this
phenomenon first in \cite{RybNON2010} under additional technical assumptions.
We eventually weeded them all out in \cite{GruRybSIMA15} when the full power
of the theory of Hankel operators was unleashed. In this sense the condition
(\ref{cond 1}) is not new but we present here a better proof.

Our condition (\ref{cond 2}) is new. It apparently improves $N$ in
(\ref{Cohen Kappeler cond 2}) by $1/4$. We can actually show that $N=11/4$
cannot be improved within the Cohen-Kappeler approach. We save extra $1/4$ by
representing the symbol of our Hankel operator (the Marchenko operator in
disguise) in a suitable form. This representation is very natural and common
in the theory of Hankel operators but is obscured in the Marchenko form. It
then invites the famous characterization of trace class Hankel operators due
to Peller \cite{Peller2003}. We first noticed the relevance of Peller's
theorem in \cite{RybNON2011} but were able to overcome numerous technical
difficulties only recently in \cite{RybPAMS18}, \cite{GruRybMatNotes18}. We
could not however achieve the condition (\ref{cond 2}) and in fact could not
even beat $N=11/4$. This is done in the current paper by finding a new
representation of the reflection coefficient, Proposition \ref{on R}. Thus
Proposition \ref{on R} combined with Theorem \ref{Trace class theorem} taken
from our \cite{GruRybMatNotes18} leads to the condition (\ref{cond 2}).

What we find remarkable is that Theorem \ref{MainThm} comes with an explicit
determinant formula for our solution (an extension of the Dyson formula). We
postpone its discussion till Section \ \ref{main} when we have all necessary terminology.

Theorem \ref{MainThm} immediately implies

\begin{theorem}
Suppose that $q$ in (\ref{KdV}) is real,
\begin{equation}%
{\displaystyle\sum\limits_{n=-\infty}^{\infty}}
\left(  \int_{n}^{n+1}\left\vert q\left(  x\right)  \right\vert dx\right)
^{2}<\infty, \label{cond 1'}%
\end{equation}
and%
\[
\int^{\infty}\left(  1+\left\vert x\right\vert \right)  ^{N}\left\vert
q\left(  x\right)  \right\vert dx<\infty,\ \ \ N\geq5/2,
\]
then the problem (\ref{KdV}) has a unique classical solution $u\left(
x,t\right)  $ such that%
\begin{equation}
\lim_{t\rightarrow+0}u(x,t)=q\left(  x\right)  \text{ in }H^{-1}\left(
\mathbb{R}\right)  . \label{IC in H^-1}%
\end{equation}

\end{theorem}

Indeed, since the condition (\ref{cond 1'}) clearly implies (\ref{cond 1}) and
hence Theorem \ref{MainThm} applies, we have a classical solution $u\left(
x,t\right)  $. On the other hand, (\ref{cond 1'}) also means that $q\in
H^{-1}\left(  \mathbb{R}\right)  $ and hence, due to the well-posedness in
$H^{-1}$ (see \cite{Killip2018}), (\ref{IC in H^-1}) holds. The convergence
(\ref{convergence}) is then superfluous as it merely follows from the well-posedness.

In fact, (\ref{cond 1'}) can be replaced with $q\in H^{-1}\left(
\mathbb{R}\right)  $. The arguments follow our \cite{GruRemRyb2015} where we
treat $H_{\operatorname*{loc}}^{-1}\left(  \mathbb{R}\right)  $ initial data
supported on a left half line. We leave the full proof out.

Note that Theorem \ref{MainThm} does not require specifying in what sense the
initial condition is understood. In fact, we do not rule out the existence of
a different solution to (\ref{KdV}) but such a solution will not be physical
as the natural requirement (\ref{convergence}) is clearly lost. In
\cite{RybNON2011}, under some additional condition we show that
(\ref{IC in H^-1}) holds in $L^{2}\left(  a,\infty\right)  $ for any
$a>-\infty$. We believe our Hankel operator approach offers some optimal
statements about initial condition. We plan to address it elsewhere.

Note that our theorems demonstrate a strong smoothing effect of the KdV flow
(see section \ref{conclusions}).

The paper is organized as follows. The short Section \ref{notation} is devoted
to our agreement on notation. In Section \ref{Refl} we present some background
on scattering theory and establish some properties of the reflection
coefficient crucially important for what follows. In Section \ref{HO} we give
brief background information on Hankel operators and prepare some statements
for the following sections. In Section \ref{sect on sep of inft} we introduce
what we maned separation of infinities principle which makes the proof of
Theorem \ref{MainThm} much more structured and easier to follow. Section
\ref{main} is devoted to the proof of Theorem \ref{MainThm} and the final
section \ref{conclusions} is reserved for relevant discussions.

\section{Notations\label{notation}}

We follow standard notation accepted in Analysis. For number sets:
$\mathbb{N}_{0}=\left\{  0,1,2,...\right\}  $, $\mathbb{R}$ is the real line,
$\mathbb{R}_{\pm}=(0,\pm\infty)$, $\mathbb{C}$ is the complex plane,
$\mathbb{C}^{\pm}=\left\{  z\in\mathbb{C}:\pm\operatorname{Im}z>0\right\}  $.
$\overline{z}$ is the complex conjugate of $z.$

Besides number sets, black board bold letters will also be used for (linear)
operators. As always, $\partial_{x}^{n}:=\partial^{n}/\partial x^{n}.$

As usual, $L^{p}\left(  S\right)  ,\ 0<p\leq\infty$, is the Lebesgue space on
a set $S$. If $S=\mathbb{R}$ then we abbreviate $\ L^{p}\left(  \mathbb{R}%
\right)  =L^{p}$. We will also deal with the weighted $L^{1}$ spaces
\[
L_{N}^{1}\left(  S\right)  =\left\{  f\ |\ \int_{S}\left(  1+\left\vert
x\right\vert ^{N}\right)  \left\vert f\left(  x\right)  \right\vert
dx<\infty\right\}  ,\ \ N>0.
\]
This function class is basic for scattering theory for 1D Schr\"{o}dinger operators.

\section{The structure of the reflection coefficient\label{Refl}}

Through this section we assume that $q$ is short-range, i.e. $q\in L_{1}^{1}$.
Associate with $q$ the full line Schr\"{o}dinger operator $\mathbb{L}%
_{q}=-\partial_{x}^{2}+q(x)$. As is well-known, $\mathbb{L}_{q}$ is
self-adjoint on $L^{2}$ and its spectrum consists of a finite number of simple
negative eigenvalues $\{-\kappa_{n}^{2}\}$, called bound states, and two fold
absolutely continuous component filling $\mathbb{R}_{+}$. There is no singular
continuous spectrum. Two linearly independent (generalized) eigenfunctions of
the a.c. spectrum $\psi_{\pm}(x,k),\;k\in\mathbb{R}$, can be chosen to
satisfy
\begin{equation}
\psi_{\pm}(x,k)=e^{\pm ikx}+o(1),\;\partial_{x}\psi_{\pm}(x,k)\mp ik\psi_{\pm
}(x,k)=o(1),\ \ x\rightarrow\pm\infty. \label{eq6.2}%
\end{equation}
The functions $\psi_{\pm}$ are referred to as Jost solutions of the
Schr\"{o}dinger equation
\begin{equation}
\mathbb{L}_{q}\psi=k^{2}\psi. \label{eq6.3}%
\end{equation}
Since $q$ is real, $\overline{\psi_{\pm}}$ also solves (\ref{eq6.3}) and one
can easily see that the pairs $\{\psi_{+},\overline{\psi_{+}}\}$ and
$\{\psi_{-},\overline{\psi_{-}}\}$ form fundamental sets for (\ref{eq6.3}).
Hence $\psi_{\mp}$ is a linear combination of $\{\psi_{\pm},\overline
{\psi_{\pm}}\}$. We write this fact as follows ($k\in\mathbb{R}$)
\begin{align}
T(k)\psi_{-}(x,k)  &  =\overline{\psi_{+}(x,k)}+R(k)\psi_{+}%
(x,k),\label{R basic scatt identity}\\
T(k)\psi_{+}(x,k)  &  =\overline{\psi_{-}(x,k)}+L(k)\psi_{-}(x,k),
\label{L basic scatt identity}%
\end{align}
where $T,R,$ and $L$ are called transmission, right, and left reflection
coefficients respectively. The identities (\ref{R basic scatt identity}%
)-(\ref{L basic scatt identity}) are totally elementary but serve as a basis
for inverse scattering theory and for this reason they are commonly referred
to as basic scattering relations. As is well-known (see, e.g.
\cite{MarchBook2011}), the triple $\{R,(\kappa_{n},c_{n})\}$, where
$c_{n}=\left\Vert \psi_{+}(\cdot,i\kappa_{n})\right\Vert ^{-1}$, determines
$q$ uniquely and is called the scattering data for $\mathbb{L}_{q}$. We will need

\begin{proposition}
[Structure of the classical reflection coefficient]\label{on R}Suppose $q$ is
real and in $L_{1}^{1}$ and $q_{\pm}=\left.  q\right\vert _{\mathbb{R}_{\pm}}$
is the restriction of $q$ to $\mathbb{R}_{\pm}$. Let $\{R,(\kappa_{n}%
,c_{n})\},$ $\{R_{+},(\kappa_{n}^{+},c_{n}^{+})\}$ be the scattering data for
$\mathbb{L}_{q},\mathbb{L}_{q_{+}}$ respectively. Then%
\begin{equation}
R=G+R_{+}. \label{R-split}%
\end{equation}
The function $G$ admits the representation%
\begin{equation}
G=\frac{T_{+}^{2}R_{-}}{1-L_{+}R_{-}}, \label{GG}%
\end{equation}
where $T_{+},L_{+}$ are the transmission and the left reflection coefficients
from $q_{+}$ and $R_{-}$ is the right reflection coefficient from $q_{-}$. The
function $G$ is bounded on $\mathbb{R}$ and meromorphic on $\mathbb{C}^{+}$
with simple poles at $\left(  i\kappa_{n}\right)  $ and $\left(  i\kappa
_{n}^{+}\right)  $ with residues%
\begin{equation}
\operatorname*{Res}_{k=i\kappa_{n}}G(k)=ic_{n},\ \ \ \operatorname*{Res}%
_{k=i\kappa_{n}^{+}}G(k)=ic_{n}^{+}, \label{residues}%
\end{equation}
Furthermore,%
\begin{equation}
R_{+}\left(  k\right)  =T_{+}\left(  k\right)  \left\{  \frac{1}{2ik}\int%
_{0}^{\infty}e^{-2ikx}q\left(  x\right)  dx+\frac{1}{\left(  2ik\right)  ^{2}%
}\int_{0}^{\infty}e^{-2ikx}Q^{\prime}\left(  x\right)  dx\right\}  ,
\label{rep for R+}%
\end{equation}
where $Q$ is an absolutely continuous function subject to%
\begin{equation}
\left\vert Q^{\prime}\left(  x\right)  \right\vert \leq C_{1}\left\vert
q\left(  x\right)  \right\vert +C_{2}\int_{x}^{\infty}\left\vert q\right\vert
,\ \ x\geq0, \label{props of Q}%
\end{equation}
with some (finite) constants $C_{1},C_{2}$ dependent on $\left\Vert
q_{+}\right\Vert _{L^{1}}$ and $\left\Vert q_{+}\right\Vert _{L_{1}^{1}}$ only.
\end{proposition}

\begin{proof}
From (\ref{R basic scatt identity}) we have%
\[%
\begin{array}
[c]{ccc}%
R(k) & = & T(k)\frac{\psi_{-}\left(  0,k\right)  }{\psi_{+}\left(  0,k\right)
}-\frac{\overline{\psi_{+}\left(  0,k\right)  }}{\psi_{+}\left(  0,k\right)
}\\
R_{+}(k) & = & \frac{T_{+}(k)}{\psi_{+}\left(  0,k\right)  }-\frac
{\overline{\psi_{+}\left(  0,k\right)  }}{\psi_{+}\left(  0,k\right)  }%
\end{array}
.
\]
Subtracting these equations yields%
\[
R(k)=R_{+}(k)+G\left(  k\right)  ,
\]
where%
\begin{equation}
G\left(  k\right)  :=T(k)\frac{\psi_{-}\left(  0,k\right)  }{\psi_{+}\left(
0,k\right)  }-\frac{T_{+}(k)}{\psi_{+}\left(  0,k\right)  } \label{G}%
\end{equation}
We refer to our \cite{RybNON2011} for the details of derivation of (\ref{GG}).
The function $G$, initially defined and bounded on the real line, can be
analytically continued into $\mathbb{C}^{+}$ (since $T$ is meromorphic in
$\mathbb{C}^{+}$ and $\psi_{\pm}$ are analytic there). Its singularities
(including removable) come apparently from the poles of $T,T_{+}$ and the
zeros of $\psi_{+}\left(  0,k\right)  $. It is well-known from the classical
1D scattering theory (see, e.g. \cite{Deift79}) that the poles of $T,T_{+}$
occur at $\left(  i\kappa_{n}\right)  $, $\left(  i\kappa_{n}^{+}\right)  $,
where $\left(  -\kappa_{n}^{2}\right)  $, $\left(  -\left(  \kappa_{n}%
^{+}\right)  ^{2}\right)  $ are the (negative) bound states of $\mathbb{L}%
_{q}$ and $\mathbb{L}_{q_{+}}$ respectively and moreover,%
\[
\operatorname*{Res}_{k=i\kappa_{n}}T(k)\frac{\psi_{-}(0,k)}{\psi_{+}%
(0,k)}=ic_{n},\ \ \ \operatorname*{Res}_{k=i\kappa_{n}^{+}}\frac{T_{+}%
(k)}{\psi_{+}(0,k)}=ic_{n}^{+}.
\]
This combined with (\ref{G}) implies (\ref{residues}). We now show that zeros
of $\psi_{+}\left(  0,k\right)  $ are removable singularities of $G$. It
follows from (\ref{R basic scatt identity}) that
\begin{equation}
T=\frac{2ik}{W(\psi_{-},\psi_{+})},\ \ T_{+}\left(  k\right)  =\frac
{2ik}{W(\psi_{0,-},\psi_{0,+})},\ \label{eq6.8}%
\end{equation}
where $\psi_{0,\pm}$ are the Jost solutions corresponding to $q_{+}$ and
$W\left(  f,g\right)  =fg^{\prime}-f^{\prime}g$ stands for the Wronskian. For
$G$ we then have%
\[
G\left(  k\right)  =\frac{2ik}{\psi_{+}\left(  0,k\right)  }\left\{
\frac{\psi_{-}\left(  0,k\right)  }{W(\psi_{-},\psi_{+})}-\frac{1}%
{W(\psi_{0,-},\psi_{0,+})}\right\}  .
\]
Since $\psi_{0,-}\left(  x,k\right)  =e^{-ikx},\ x\leq0,$ and $\psi
_{0,+}\left(  x,k\right)  =\psi_{+}\left(  x,k\right)  ,\ x\geq0,$ one
concludes that ($W$ is independent of $x$)%
\begin{align}
W(\psi_{0,-},\psi_{0,+})  &  =\psi_{0,-}\left(  0,k\right)  \partial_{x}%
\psi_{0,+}\left(  0,k\right)  -\partial_{x}\psi_{0,-}\left(  0,k\right)
\psi_{0,+}\left(  0,k\right) \nonumber\\
&  =\partial_{x}\psi_{+}\left(  0,k\right)  +ik\psi_{+}\left(  0,k\right)  ,
\label{W}%
\end{align}
and we arrive at%
\[
G\left(  k\right)  =\frac{2ik}{W(\psi_{-},\psi_{+})}\frac{\partial_{x}\psi
_{-}\left(  0,k\right)  +ik\psi_{-}\left(  0,k\right)  }{\partial_{x}\psi
_{+}\left(  0,k\right)  +ik\psi_{+}\left(  0,k\right)  }.
\]
It now follows from (\ref{W}) that a zero of $\psi_{+}\left(  0,k\right)  $
cannot be a zero of $\partial_{x}\psi_{+}\left(  0,k\right)  $ (otherwise
$\psi_{0,-}$ and $\psi_{0,+}$ were linearly dependant) and thus a zero of
$\psi_{+}\left(  0,k\right)  $ is not a pole of $G\left(  k\right)  $.

Turn now to (\ref{rep for R+}). To this end we use the following
representation from \cite{Deift79}
\begin{equation}
2ik\ \frac{R_{+}(k)}{T_{+}(k)}=%
{\displaystyle\int_{-\infty}^{\infty}}
g(y)e^{-2iky}dy, \label{R/T}%
\end{equation}
where $g$ is defined as follows. Let%
\[
y_{\pm}\left(  x,k\right)  =e^{\mp ikx}\psi_{\pm}(x,k).
\]
As is shown in \cite{Deift79}, $y_{\pm}\left(  x,k\right)  -1\in H^{2}$ for
every $x$,
\[
y_{\pm}(x,k)=1\pm\int_{0}^{\pm\infty}B_{\pm}(x,y)e^{\pm2iky}\,dy,
\]
(i.e. the Fourier representation of $y_{\pm}\left(  x,k\right)  -1$) and
\begin{align*}
g(y)  &  =-\partial_{x}B_{+}(0,y)+\partial_{x}B_{-}(0,y)+\int\partial_{x}%
B_{-}(0,z)B_{+}(0,y-z)dz\\
&  -\int\partial_{x}B_{+}(0,z)B_{-}(0,x-z)dz.
\end{align*}
In our case $y_{-}\left(  x,k\right)  =1$ for $x\leq0$ and hence
$B_{-}(x,y)=0$. Therefore the previous equation simplifies to%
\begin{equation}
g(y)=-\partial_{x}B_{+}(0,y). \label{g}%
\end{equation}
$B_{+}(x,y)$, in turn, solves the integral equation \cite{Deift79}
\[
B_{+}(x,y)-\int_{0}^{y}\left(  \int_{x+y-z}^{\infty}q_{+}(t)B_{+}%
(t,z)dt\right)  dz=\int_{x+y}^{\infty}q_{+}(t)dt,\;y\geq0.
\]
Differentiating this equation in $x$ and setting $x=0$ yields%
\begin{align}
g(y)  &  =q_{+}(y)+\int_{0}^{y}q_{+}(y-z)B_{+}(y-z,z)dz\nonumber\\
&  =q_{+}(y)+\int_{0}^{y}q_{+}(z)B_{+}(z,y-z)dz\nonumber\\
&  =:q(y)+Q\left(  y\right)  . \label{q+Q}%
\end{align}
Let us now study $Q$. It is clearly supported on $\left(  0,\infty\right)  $
and one has
\begin{equation}
Q^{\prime}(y)=q(y)B_{+}(y,0)+\int_{0}^{y}q(z)\partial_{y}B_{+}(z,y-z)dz:=g_{1}%
(y)+g_{2}(y). \label{Q'}%
\end{equation}
To obtain the desired estimate (\ref{props of Q}) we make use of two crucially
important estimates from \cite{Deift79}: for $q\in L_{1}^{1}$%
\begin{equation}
\left\vert B_{+}(x,y)\right\vert \leq\eta(x+y)e^{\gamma(x)}, \label{eq1.5}%
\end{equation}
and
\begin{equation}
\left\vert \partial_{y}B_{+}(x,y)+q(x+y)\right\vert \leq2\eta(x+y)\eta
(x)e^{\gamma(x)}, \label{eq1.6}%
\end{equation}
where
\[
\gamma(x)=\int_{x}^{\infty}(t-x)|q(t)|dt,\;\eta(x)=\int_{x}^{\infty}|q(t)|dt.
\]
Since for $x\geq0$
\[
\gamma(x)\leq\int_{x}^{\infty}t|q(t)|dt\leq\int_{0}^{\infty}t|q(t)|dt=\gamma
\left(  0\right)  ,
\]
it follows from (\ref{eq1.5})-(\ref{eq1.6}) that (recalling that $y\geq0$)%
\begin{align}
|g_{1}(y)|  &  \leq|q(y)|\,\eta(y)e^{\gamma(y)}\label{g1}\\
&  \leq\ \ \eta(0)e^{\gamma\left(  0\right)  }\ \left\vert \,q(y)\right\vert
\nonumber
\end{align}
and%
\begin{align}
|g_{2}(y)|  &  \leq\left\vert q(y)\right\vert \int_{0}^{y}\left\vert
q\right\vert +2\eta(y)\int_{0}^{y}|q(z)|\eta(z)e^{\gamma(z)}dz\label{g2}\\
&  \leq\left\vert q(y)\right\vert \int_{0}^{y}\left\vert q\right\vert
+2\eta^{2}(0)e^{\gamma\left(  0\right)  }\ \eta(y).\nonumber
\end{align}
Combining now \ (\ref{Q'}) and (\ref{g1})-(\ref{g2}) yields (\ref{props of Q}).

It remains to show (\ref{rep for R+}). Substituting (\ref{q+Q}) into
(\ref{R/T}) we have%
\[
2ik\ \frac{R_{+}(k)}{T_{+}(k)}=%
{\displaystyle\int_{0}^{\infty}}
q(y)e^{-2iky}dy+%
{\displaystyle\int_{0}^{\infty}}
Q(y)e^{-2iky}dy.
\]
Evaluating the last integral by parts yields%
\[%
{\displaystyle\int_{0}^{\infty}}
Q(y)e^{-2iky}dy=-\left.  \frac{Q(y)e^{-2i\lambda y}}{2i\lambda}\right\vert
_{0}^{\infty}+\frac{1}{2i\lambda}\int_{0}^{\infty}Q^{\prime}(y)e^{-2i\lambda
y}dy.
\]
It follows from (\ref{q+Q}) and (\ref{eq1.5}) that the integrated term
vanishes and (\ref{rep for R+}) is proven.
\end{proof}

The split (\ref{R-split}) implies that the right reflection coefficient $R$
can be represented as an analytic function plus the right reflection
coefficient $R_{+}$ which need not admit analytic continuation from the real
line. Moreover, $R_{+}$ is completely determined by $q$ on $\left(
0,\infty\right)  $ (by simple shifting arguments, any interval $\left(
a,\infty\right)  $ can be considered). Some parts of Proposition \ref{on R}
appeared in our \cite{RybNON2011} and \cite{GruRybSIMA15}) but
(\ref{rep for R+}) is new. For $q$ supported on the full line, it was proven
in \cite{Deift79} that%
\[
R\left(  k\right)  =\frac{T\left(  k\right)  }{2ik}\int_{-\infty}^{\infty
}e^{-2ikx}g\left(  x\right)  dx,
\]
where $g$ satisfies%
\begin{equation}
\left\vert g\left(  x\right)  \right\vert \leq\left\vert q\left(  x\right)
\right\vert +const\left\{
\begin{array}
[c]{ccc}%
\int_{x}^{\infty}\left\vert q\right\vert  & , & x\geq0\\
\int_{-\infty}^{x}\left\vert q\right\vert  & , & x<0
\end{array}
\right.  , \label{|g|}%
\end{equation}
and nothing better can be said about $g$ in general. In the case of $q$
supported on $\left(  0,\infty\right)  $ this statement can be improved.
Indeed, (\ref{rep for R+}) implies that%
\[
g\left(  x\right)  =q\left(  x\right)  +Q\left(  x\right)
\]
with some absolutely continuous on $\left(  0,\infty\right)  $ function which
derivative $Q^{\prime}$ satisfies (\ref{|g|}).

\section{Hankel operators with oscillatory symbols\label{HO}}

We refer the reader to \cite{Nik2002} and \cite{Peller2003} for background
reading on Hankel operators. We recall that a function $f$ analytic in
$\mathbb{C}^{\pm}$ is in the Hardy space $H_{\pm}^{2}$ if
\[
\sup_{y>0}\int_{-\infty}^{\infty}\left\vert f(x\pm iy)\right\vert
^{2}\ dx<\infty.
\]
We will also need $H_{\pm}^{\infty}$, the algebra of analytic functions
uniformly bounded in $\mathbb{C}^{\pm}$. It is particularly important that
$H_{\pm}^{2}$ is a Hilbert space with the inner product induced from $L^{2}$:
\[
\langle f,g\rangle_{H_{\pm}^{2}}=\langle f,g\rangle_{L^{2}}=\left\langle
f,g\right\rangle =\int_{-\infty}^{\infty}f\left(  x\right)  \bar{g}\left(
x\right)  dx.
\]
It is well-known that $L^{2}=H_{+}^{2}\oplus H_{-}^{2},$ the orthogonal
(Riesz) projection $\mathbb{P}_{\pm}$ onto $H_{\pm}^{2}$ being given by%
\begin{equation}
(\mathbb{P}_{\pm}f)(x)=\pm\frac{1}{2\pi i}\int_{-\infty}^{\infty}%
\frac{f(s)\ ds}{s-(x\pm i0)}. \label{proj}%
\end{equation}

Let $(\mathbb{J}f)(x)=f(-x)$ be the operator of reflection. Given $\varphi\in
L^{\infty}$ the operator $\mathbb{H}(\varphi):H_{+}^{2}\rightarrow H_{+}^{2}$
defined by the formula%
\begin{equation}
\mathbb{H}(\varphi)f=\mathbb{JP}_{-}\varphi f,\ \ \ f\in H_{+}^{2},
\label{Hankel}%
\end{equation}
is called the Hankel\emph{ }operator with symbol $\varphi$.

It directly follows from the definition (\ref{Hankel}) that the Hankel
operator $\mathbb{H}(\varphi)$ is bounded if its symbol $\varphi$ is bounded
and $\mathbb{H}(\varphi+h)=\mathbb{H}(\varphi)$ for any $h\in H_{+}^{\infty}$.
The latter means that only part of $\varphi$ analytic in $\mathbb{C}^{-}$
(called co-analytic) matters. More specifically,%
\[
\mathbb{H}(\varphi)=\mathbb{H}(\widetilde{\mathbb{P}}_{-}\varphi),
\]
where%
\begin{align}
(\widetilde{\mathbb{P}}_{-}\varphi)(x)  &  =-\frac{1}{2\pi i}\int_{-\infty
}^{\infty}\left(  \frac{1}{\lambda-(x-i0)}-\frac{1}{\lambda+i}\right)
\varphi(\lambda)d\lambda\label{P-}\\
&  =(x+i)\left(  \mathbb{P}_{-}\frac{1}{\cdot+i}\varphi\right)
(x),\;\ \ \varphi\in L^{\infty}.\nonumber
\end{align}
We note that in general $\widetilde{\mathbb{P}}_{-}\varphi\notin H_{-}%
^{\infty}$ if $\varphi\in L^{\infty}$ but the Hankel operator $\mathbb{H}%
(\varphi)$ is still well-defined by (\ref{Hankel}) and bounded. If $\varphi\in
L^{2}$ then $\widetilde{\mathbb{P}}_{-}\varphi$ differs from $\mathbb{P}%
_{-}\varphi$ by a constant and thus $\mathbb{P}_{-}\varphi$ can be take as the
co-analytic part.

In the context of the KdV equation symbols of the following form%
\[
\ \varphi\left(  x\right)  =G\left(  x\right)  \xi_{\alpha,\beta}\left(
x\right)  ,
\]
naturally arise. Here $G\in L^{\infty}$, and%
\[
\text{ }\xi_{\alpha,\beta}\left(  x\right)  =\exp i\left(  \alpha x+\beta
x^{3}\right)  ,
\]
where $\alpha,\beta$ are real parameters, and $\beta>0.$ The main feature of
$\xi_{\alpha,\beta}$ is a rapid decay along any line $\mathbb{R}+ih$ in the
upper half plane and as a result the quality of $\mathbb{H}\left(
G\xi_{\alpha,\beta}\right)  $ may actually be better than $\mathbb{H}\left(
G\right)  $. E.g., if $G\in L^{\infty}$ and is analytic in $\mathbb{C}^{+}$
then (\ref{P-}) takes form ($h>0$)
\begin{equation}
(\widetilde{\mathbb{P}}_{-}\varphi)(x)=-\frac{1}{2\pi i}\int_{\mathbb{R}%
+ih}\left(  \frac{1}{\lambda-x}-\frac{1}{\lambda+i}\right)  G\left(
\lambda\right)  \xi_{\alpha,\beta}\left(  \lambda\right)  d\lambda
,\label{P- h}%
\end{equation}
which is an entire function as long as this integral is absolutely convergent.
This means that $\mathbb{H}\left(  G\xi_{\alpha,\beta}\right)  $ is in any
Shatten-von Neumann ideal $\mathfrak{S}_{p}$ ($0<p\leq\infty$) while
$\mathbb{H}\left(  G\right)  $ need not be even compact. Better yet,
$\mathbb{H}\left(  G\xi_{\alpha,\beta}\right)  $ can be differentiated in any
$\mathfrak{S}_{p}$ norm with respect to $\alpha,\beta$ infinitely many time.
Indeed, since for all $m,n$
\[
\partial_{\alpha}^{m}\partial_{\beta}^{n}(\widetilde{\mathbb{P}}_{-}%
\varphi)(x)=-\frac{1}{2\pi i}\int_{\mathbb{R}+ih}\left(  \frac{1}{\lambda
-x}-\frac{1}{\lambda+i}\right)  G\left(  \lambda\right)  \partial_{\alpha}%
^{m}\partial_{\beta}^{n}\xi_{\alpha,\beta}\left(  \lambda\right)  d\lambda
\]
are entire functions the operators defined by
\begin{equation}
\partial_{\alpha}^{m}\partial_{\beta}^{n}\mathbb{H}\left(  G\xi_{\alpha,\beta
}\right)  =\mathbb{H}\left(  \partial_{\alpha}^{m}\partial_{\beta}%
^{n}\widetilde{\mathbb{P}}_{-}G\xi_{\alpha,\beta}\right)  \label{diff H}%
\end{equation}
are all in $\mathfrak{S}_{p}$. Note that if we formally set%
\[
\partial_{\alpha}^{m}\partial_{\beta}^{n}\mathbb{H}\left(  G\xi_{\alpha,\beta
}\right)  =\mathbb{H}\left(  \partial_{\alpha}^{m}\partial_{\beta}^{n}%
G\xi_{\alpha,\beta}\right)  ,
\]
then we would have the Hankel operator with an unbounded symbol $\left(
ix\right)  ^{m+3n}G\left(  x\right)  \xi_{\alpha,\beta}\left(  x\right)  $.
Thus, (\ref{diff H}) can be viewed as a way to regularize Hankel operators
with certain unbounded oscillatory symbols.

We have to work a bit harder if $G$ doesn't extend analytically into
$\mathbb{C}^{+}$ but has some smoothness. We can no longer apply the Cauchy
theorem to evaluate $\widetilde{\mathbb{P}}_{-}\varphi$ but the Cauchy-Green
formula will do. This is the case when
\[
G\left(  x\right)  =\int_{0}^{\infty}e^{-ixs}g\left(  s\right)  ds
\]
with some $g\in L_{N}^{1}\left(  \mathbb{R}_{+}\right)  $, $N\geq1$.
Apparently for any integer $n\leq N$%
\begin{equation}
G^{\left(  n\right)  }\in H^{\infty}\left(  \mathbb{C}^{-}\right)  \cap
C_{0}\left(  \mathbb{R}\right)  \label{n deriv}%
\end{equation}
but $G$ doesn't in general extend analytically into $\mathbb{C}^{+}$ and we
can no longer deform the contour into the upper half plane. Let us now
consider instead its pseudoanalytic extension into $\mathbb{C}^{+}$. Following
\cite{Dyn76} we call $F\left(  x,y\right)  $ a pseudoanalytic extension of
$f\left(  x\right)  $ into $\mathbb{C}$ if
\[
F\left(  x,0\right)  =f\left(  x\right)  \text{ and }\overline{\partial
}F\left(  x,y\right)  \rightarrow0,y\rightarrow0,
\]
where $\overline{\partial}:=\left(  1/2\right)  \left(  \partial_{x}%
+i\partial_{y}\right)  $. Note that due to (\ref{n deriv}) for $n\leq N$ the
Taylor formula%
\begin{equation}
G\left(  z,\overline{z}\right)  =\sum_{m=0}^{n-1}\frac{G^{\left(  m\right)
}\left(  \overline{z}\right)  }{m!}\left(  z-\overline{z}\right)
^{m},\ \ \ z\in\mathbb{C}^{+},\label{taylor}%
\end{equation}
defines such continuation as $G\left(  z,\overline{z}\right)  $ clearly agrees
with $G$ on the real line and for $\lambda\in\mathbb{C}^{+}$%
\begin{equation}
\overline{\partial}G\left(  z,\overline{z}\right)  =\frac{G^{\left(  n\right)
}\left(  \overline{z}\right)  }{\left(  n-1\right)  !}\left(  z-\overline
{z}\right)  ^{n-1},\ \ \ n\leq N.\label{decay close to R}%
\end{equation}
By the Cauchy-Green formula applied, say, to the strip $0\leq\operatorname{Im}%
z\leq1$ we have ($\lambda=u+iv$)%
\begin{align}
&  \widetilde{\mathbb{P}}_{-}G\xi_{\alpha,\beta}\left(  x\right)  \label{3}\\
&  =\frac{x+i}{2\pi i}\int_{-\infty}^{\infty}\frac{\xi_{\alpha,\beta}\left(
\lambda\right)  G\left(  \lambda\right)  }{\lambda+i}\ \frac{d\lambda}%
{\lambda-\left(  x-i0\right)  }\nonumber\\
&  =\frac{x+i}{2\pi i}\int_{\mathbb{R}+i}\frac{\xi_{\alpha,\beta}\left(
\lambda\right)  G\left(  \lambda,\overline{\lambda}\right)  }{\lambda
+i}\ \frac{d\lambda}{\lambda-x}\nonumber\\
&  +\frac{x+i}{\pi}\int_{0\leq\operatorname{Im}\lambda\leq1}\frac{\xi
_{\alpha,\beta}\left(  \lambda\right)  \overline{\partial}G\left(
\lambda,\overline{\lambda}\right)  }{\lambda+i}\ \frac{dudv}{\lambda
-x}.\nonumber
\end{align}
The first integral on the right hand side of (\ref{3}) is identical to
(\ref{P- h}) and thus we only need to study%
\[
\phi_{\alpha,\beta}\left(  x\right)  :=\frac{x+i}{\pi}\int_{0\leq
\operatorname{Im}\lambda\leq1}\frac{\xi_{\alpha,\beta}\left(  \lambda\right)
\overline{\partial}G\left(  \lambda,\overline{\lambda}\right)  }{\lambda
+i}\ \frac{dudv}{\lambda-x},
\]
where%
\begin{align*}
\overline{\partial}G\left(  \lambda,\overline{\lambda}\right)   &
=\frac{G^{\left(  n\right)  }\left(  \overline{\lambda}\right)  }{\left(
n-1\right)  !}\left(  \lambda-\overline{\lambda}\right)  ^{n-1}\\
&  =\left\{  \int_{0}^{\infty}\left(  2s\right)  ^{n}e^{-i\overline{\lambda}%
s}g\left(  s\right)  ds\right\}  \frac{v^{n-1}}{2i\left(  n-1\right)  !}.
\end{align*}
We have%
\begin{align*}
&  \phi_{\alpha,\beta}\left(  x\right)  \\
&  =\frac{x+i}{2\pi i}\int_{0\leq\operatorname{Im}\lambda\leq1}\frac
{\xi_{\alpha,\beta}\left(  \lambda\right)  }{\lambda+i}\ \left\{  \int%
_{0}^{\infty}\left(  2s\right)  ^{n}e^{-i\overline{\lambda}s}g\left(
s\right)  ds\right\}  \frac{v^{n-1}}{\left(  n-1\right)  !}\frac{dudv}%
{\lambda-x}\\
&  =\frac{x+i}{2\pi i}\int_{0}^{\infty}\left\{  \int_{0}^{1}dve^{-2vs}%
v^{n-1}\int_{\mathbb{R}+iv}\frac{\xi_{\alpha,\beta-s}\left(  \lambda\right)
}{\lambda+i}\frac{d\lambda}{\lambda-x}\right\}  \left(  2s\right)
^{n}g\left(  s\right)  ds.
\end{align*}
The integral with respect to $d\lambda$ is clearly independent of contour and
hence%
\begin{equation}
\phi_{\alpha,\beta}\left(  x\right)  =\int_{0}^{\infty}I_{0}\left(
s,\alpha,\beta\right)  \gamma_{n}\left(  s\right)  \left(  2s\right)
^{n}g\left(  s\right)  ds,\label{phi alfa better}%
\end{equation}
where%
\[
I_{0}\left(  x,\alpha-s,\beta\right)  :=\frac{x+i}{2\pi i}\int_{\mathbb{R}%
+i}\frac{\xi_{\alpha-s,\beta}\left(  \lambda\right)  }{\lambda+i}%
\frac{d\lambda}{\lambda-x}%
\]
and%
\[
\gamma_{n}\left(  s\right)  :=\int_{0}^{1}dve^{-2vs}\frac{v^{n-1}}{\left(
n-1\right)  !}.
\]
Differentiating $\phi_{\alpha,\beta}\left(  x\right)  $ formally in $\beta$ we
have%
\begin{equation}
\partial_{\alpha}^{j}\phi_{\alpha,\beta}\left(  x\right)  =\int_{0}^{\infty
}\partial_{\alpha}^{j}I_{0}\left(  x,\alpha-s,\beta\right)  \left(  2s\right)
^{n}\gamma_{n}\left(  s\right)  g\left(  s\right)  ds.\label{d_b}%
\end{equation}
Apparently, this formal differentiation is valid as long as the integral is
absolutely convergent. But%
\begin{align*}
\partial_{\alpha}^{j}I_{0}\left(  x,\alpha-s,\beta\right)   &  =\frac
{x+i}{2\pi i}\int_{\mathbb{R}+i}\frac{\left(  i\lambda\right)  ^{j}\xi
_{\alpha-s,\beta}\left(  \lambda\right)  }{\lambda+i}\frac{d\lambda}%
{\lambda-x}\\
&  =:I_{j}\left(  x,\alpha-s,\beta\right)
\end{align*}
is clearly absolutely convergent and%
\[
\left(  2s\right)  ^{n}\gamma_{n}\left(  s\right)  \leq\left(  2s\right)
^{n}\int_{0}^{\infty}dve^{-2vs}\frac{v^{n-1}}{\left(  n-1\right)  !}=1.
\]
Note that the integral defining $I_{j}\left(  x,\alpha-s,\beta\right)  $ is
independent of contour. The current one, $\mathbb{R}+i$, is not suitable for
getting required bounds on its growth in $s\rightarrow\infty$ and we will
later deform it as needed (see (\ref{integral})). It follows from (\ref{d_b})
that%
\[
\left\vert \partial_{\alpha}^{j}\phi_{\alpha,\beta}\left(  x\right)
\right\vert \leq\int_{0}^{\infty}\left\vert I_{j}\left(  x,\alpha
-s,\beta\right)  \right\vert \left\vert g\left(  s\right)  \right\vert ds
\]
and thus $\partial_{\alpha}^{j}\mathbb{H}\left(  \phi_{\alpha,\beta}\right)  $
is well-defined by $\partial_{\alpha}^{j}\mathbb{H}\left(  \phi_{\alpha,\beta
}\right)  =\mathbb{H}\left(  \partial_{\alpha}^{j}\phi_{\alpha,\beta}\right)
$ as a bounded operator if for each $\alpha$ and $\beta>0$
\[
\int_{-\beta}^{\infty}\left\vert I_{j}\left(  x,-s,\beta\right)  \right\vert
\left\vert g\left(  s+\beta\right)  \right\vert ds\in L^{\infty}.
\]
We will however need conditions on the decay of $g$ which guarantee the
membership of $\partial_{\alpha}^{j}\mathbb{H}\left(  \phi_{\alpha,\beta
}\right)  $ in trace class $\mathfrak{S}_{1}$ for a specified number $j$. We
studied this question in \cite{GruRybMatNotes18} where we proved

\begin{theorem}
\label{Trace class theorem}Let real $g\in L_{N}^{1}\left(  \mathbb{R}%
_{+}\right)  $ and $\phi_{\alpha,\beta}$ be given by (\ref{phi alfa better})
then the Hankel operator $\mathbb{H}\left(  \phi_{\alpha,\beta}\right)  $ is
$\left\lfloor 2N\right\rfloor -1$ times continuously differentiable in
$\alpha$ in trace norm for every real $\alpha$ and $\beta>0$.
\end{theorem}

Note that since $\partial_{\beta}^{j}\xi_{\alpha,\beta}=\partial_{\alpha}%
^{3j}\xi_{\alpha,\beta}$, Theorem \ref{Trace class theorem} can be restated
for $\beta$ accordingly. We refer to \cite{GruRybMatNotes18} for the complete
proof. We only mention that our arguments rely on a deep characterization of
trace class Hankel operators by Peller \cite{Peller2003} which says that,
given $\varphi\in L^{\infty}\left(  \mathbb{R}\right)  $, the Hankel operator
$\mathbb{H}(\varphi)$ is trace class iff $\left(  \widetilde{\mathbb{P}}%
_{-}\varphi\right)  ^{\prime\prime}\in L^{1}\left(  \mathbb{C}^{-}\right)  $
and $\sup_{\operatorname{Im}z\leq-1}\left\vert \widetilde{\mathbb{P}}%
_{-}\varphi\left(  z\right)  \right\vert <\infty$. In our case the problem
boils down to the following question. Given integer $n$, find the least
possible $N$ such that%
\begin{equation}
g\in L_{N}^{1}\left(  \mathbb{R}_{+}\right)  \Longrightarrow\int_{0}^{\infty
}\left\{  \int_{\mathbb{R}+i}\frac{\lambda^{n}\xi_{\alpha-s,\beta}\left(
\lambda\right)  }{\left(  \lambda-z\right)  ^{3}}d\lambda\right\}  g\left(
s\right)  ds\in L^{1}\left(  \mathbb{C}^{-}\right)  . \label{L1}%
\end{equation}
Proving (\ref{L1}) reduces essentially to analyzing%
\begin{equation}
\int_{0}^{\infty}dss^{n/2}\left\vert g\left(  s+\alpha\right)  \right\vert
\int_{\mathbb{C}^{+}}\left\vert \int_{\Gamma}\ e^{is^{3/2}f\left(
\lambda\right)  }\frac{\lambda^{n}d\lambda}{\left(  \lambda-x+iy\right)  ^{3}%
}\right\vert dxdy, \label{integral}%
\end{equation}
where $f\left(  \lambda\right)  =\lambda^{3}/3-\lambda$ is the phase function
and $\Gamma$ is a contour passing through its stationary points $\lambda=\pm
1$. The hardest part is treating the neighborhood of points $x-iy$ close to
$\lambda=\pm1$. One needs to use the steepest decent approximation with
coalescent stationary points and poles (see \cite{Wong2001}). The payoff is
however an optimal estimate for (\ref{integral}), which in turn means that, in
a sense, Theorem \ref{Trace class theorem} is optimal.

\section{The separation of infinities principle\label{sect on sep of inft}}

Through this section we assume that our initial data $q$ is short-range. Let
$\{R,(\kappa_{n},c_{n})\}$ be the scattering data for $\mathbb{L}_{q}$.
Consider the Hankel operator $\mathbb{H}(\varphi)$ with the symbol
\begin{equation}
\varphi(k)=\sum_{n}\frac{c_{n}\xi_{x,t}(i\kappa_{n})\,}{ik+\kappa_{n}}%
+\xi_{x,t}(k)R(k), \label{fi}%
\end{equation}
were%
\[
\xi_{x,t}(k)=e^{i\left(  8k^{3}t+2kx\right)  }.
\]

\begin{theorem}
[separation of infinities principle]\label{sep of inft}Under conditions and in
notation of Proposition \ref{on R}%
\[
\mathbb{H}\left(  \varphi\right)  =\mathbb{H}\left(  \varphi_{+}\right)
+\mathbb{H}(\Phi),
\]
where%
\begin{equation}
\varphi_{+}(k)=\sum_{n}\frac{c_{n}^{+}\xi_{x,t}(i\kappa_{n}^{+})\,}%
{ik+\kappa_{n}^{+}}+\xi_{x,t}(k)R_{+}(k)\label{fi +}%
\end{equation}
and%
\[
\Phi\left(  k\right)  =-\frac{1}{2\pi i}\int_{\mathbb{R}+ih}\frac{\xi
_{x,t}(\lambda)G(\lambda)}{\lambda-k}\ d\lambda,\ \ \ h>\max\left(  \kappa
_{n}\right)  .
\]

\end{theorem}

\begin{proof}
Set%
\[
\phi\left(  k\right)  :=\sum_{n}\frac{c_{n}\xi_{x,t}(i\kappa_{n})}%
{ik+\kappa_{n}},\ \ \ \phi_{+}\left(  k\right)  :=\sum_{n}\frac{c_{n}^{+}%
\xi_{x,t}(i\kappa_{n}^{+})\,}{ik+\kappa_{n}^{+}},
\]
which are rational function with simple poles at $\left(  i\kappa_{n}\right)
,\left(  i\kappa_{n}^{+}\right)  ,$ respectively. Consider the co-analytic
part of $\xi_{x,t}G$ (as is well-known, $R\in L^{2}$) :
\[
\left(  \mathbb{P}_{-}\xi_{x,t}G\right)  (k)=-\frac{1}{2\pi i}\int%
_{\mathbb{R}}\frac{\xi_{x,t}(\lambda)G(\lambda)}{\lambda-(k-i0)}d\lambda.
\]
By Proposition \ref{on R}, $\xi_{x,t}G$ is meromorphic in $\mathbb{C}^{+}$ and
by the residue theorem we then have ($h>\max\left(  \kappa_{n}\right)  $)%
\begin{align*}
-  &  \frac{1}{2\pi i}\int_{\mathbb{R}}\frac{\xi_{x,t}(\lambda)G(\lambda
)}{\lambda-(k-i0)}d\lambda+\frac{1}{2\pi i}\int_{\mathbb{R}+ih}\frac{\xi
_{x,t}(\lambda)G(\lambda)}{\lambda-k}\ d\lambda\\
&  =-\sum_{n}\frac{ic_{n}\xi_{x,t}(i\kappa_{n})}{i\kappa_{n}-k}+\sum_{n}%
\frac{ic_{n}^{+}\xi_{x,t}(i\kappa_{n}^{+})}{i\kappa_{n}^{+}-k}\\
&  =-\phi+\phi_{+}.
\end{align*}
It follows that%
\[
\mathbb{P}_{-}\xi_{x,t}G=\Phi-\phi+\phi_{+}.
\]
By Proposition \ref{on R} then%
\begin{align*}
\mathbb{H}\left(  \varphi\right)   &  =\mathbb{H}\left(  \phi\right)
+\mathbb{H}\left(  \xi_{x,t}R\right) \\
&  =\mathbb{H}\left(  \phi\right)  +\mathbb{H}\left(  \xi_{x,t}R_{+}\right)
+\mathbb{H}\left(  \xi_{x,t}G\right) \\
&  =\mathbb{H}\left(  \phi\right)  +\mathbb{H}\left(  \xi_{x,t}R_{+}\right)
+\mathbb{H}\left(  \mathbb{P}_{-}\xi_{x,t}G\right) \\
&  =\mathbb{H}\left(  \phi\right)  +\mathbb{H}\left(  \xi_{x,t}R_{+}\right)
+\mathbb{H}\left(  \Phi-\phi+\phi_{+}\right) \\
&  =\mathbb{H}\left(  \xi_{x,t}R_{+}\right)  +\mathbb{H}\left(  \Phi\right)
+\mathbb{H}\left(  \phi_{+}\right) \\
&  =\mathbb{H}\left(  \varphi_{+}\right)  +\mathbb{H}\left(  \Phi\right)
\end{align*}
and the theorem is proven.
\end{proof}

Theorem \ref{sep of inft} can be interpreted as follows. Given scattering data
for $\mathbb{L}_{q}$, the Hankel operator $\mathbb{H}\left(  \varphi\right)  $
associated with these data is different from the one corresponding to the data
for $\mathbb{L}_{q_{+}}$ by the Hankel operator with an analytic symbol. Thus
$\mathbb{H}\left(  \varphi_{+}\right)  $ is completely determined by $q$ on
$\left(  0,\infty\right)  $. The part $\mathbb{H}(\Phi)$ depends on $q$ on the
whole line but has some nice properties (see below).

Our application of Theorem \ref{sep of inft} to the KdV equation is based on
what we call the Dyson formula (aka Bargmann or log-determinant formula). It
says that a $L_{1}^{1}$ potential $q\left(  x\right)  $ can be recovered from
the scattering data $\{R,(\kappa_{n},c_{n})\}$ by the formula%
\begin{equation}
q(x)=-2\partial_{x}^{2}\log\det\left\{  1+\mathbb{H}(\varphi_{x})\right\}
,\varphi_{x}\left(  k\right)  :=\sum_{n}\frac{c_{n}e^{-2\kappa_{n}x}%
\,}{ik+\kappa_{n}}+e^{2ikx}R(k).\label{Dyson}%
\end{equation}
where the determinant is understood in the classical Fredholm sense. 

The formula (\ref{Dyson}) has a long history. If $R=0$ (reflectionless $q$)
the Marchenko integral equation turns into a (finite) linear system and
(\ref{Dyson}) follows immediately from the Cramer rule. This idea is extended
to the general $L_{1}^{1}$ case in Faddeev's survey \cite{Faddeev}, where it
naturally appears as nothing but a different (equivalent) way of writing the
solution to the Marchenko integral equation. We first learned about
(\ref{Dyson}) from \cite{Faddeev} but Dyson in his influential \cite{Dyson76}
refers to Faddeev's \cite{Faddeev1959} available first in Russian in 1959.
Dyson links (\ref{Dyson}) to Fredholm determinants arising in random matrix
theory and it is likely why (\ref{Dyson}) is frequently associated with him.
In the context of integrable systems, (\ref{Dyson}) is revisited in 1984 by
Poppe in \cite{Poppe1984} where it is related to the famous Hirota tau
function. We have also seen (\ref{Dyson}) used in the KdV context with
references to Bargmann and Moser (i.e. it was already known back in the early
1950s). We refer the interested reader to \cite{Bornemann2010} for many other
applications of Fredholm determinants and associated numerics. 

Since the Marchenko integral operator is unitarily equivalent to
$\mathbb{H}(\varphi_{x})$, our version (\ref{Dyson}) immediately follows from
that of \cite{Faddeev}. 

As was discussed in Introduction, the KdV equation with data $q\in L_{1}^{1}$
is well-posed at least in $H^{-s}$ with $s>0$ and its solution $u(x,t)$ can be
obtained from solving the Marchenko integral equation and written as%
\[
u(x,t)=-2\partial_{x}^{2}\log\det\left\{  1+\mathbb{H}(x,t)\right\}
,\ \ \ \mathbb{H}(x,t):=\mathbb{H}\left(  \varphi\right)  ,
\]
where $\varphi$ is defined by (\ref{fi}). As we proved in
\cite{GruRybMatNotes18}, $\mathbb{H}(x,t)$ is trace class and hence
$\det\left\{  1+\mathbb{H}(x,t)\right\}  $ is well-defined in the classical
Fredholm sense. To prove the necessary smoothness we show that the condition
(\ref{cond 2}) provides five continuous $x$ derivatives of $\mathbb{H}(x,t)$
(and one in $t$). This will be done in the next section. Incidentally,
differentiability of the Fredholm determinant is also discussed in
\cite{Poppe1984} under additional smoothness assumptions on the initial data. 

Theorem \ref{sep of inft} and the well-known formula%
\[
\det\left(
\begin{array}
[c]{cc}%
A_{11} & A_{12}\\
A_{21} & A_{22}%
\end{array}
\right)  =\det A_{11}\det\left(  A_{22}-A_{21}A_{11}^{-1}A_{12}\right)  ,
\]
readily imply

\begin{theorem}
[separation of infinities principle for KdV]\label{sep of inft for kdv}The
solution to the Cauchy problem for the KdV equation (\ref{KdV}) with $q\in
L_{1}^{1}$ can be written in the following forms
\begin{align*}
u(x,t)  &  =-2\partial_{x}^{2}\log\det\left\{  1+\mathbb{H}_{+}%
(x,t)+\mathbb{H}\left(  \Phi\right)  \right\} \\
&  =u_{+}(x,t)-2\partial_{x}^{2}\log\det\left\{  1+\left[  1+\mathbb{H}%
_{+}(x,t)\right]  ^{-1}\mathbb{H}\left(  \Phi\right)  \right\} \\
&  =-2\partial_{x}^{2}\log\det\left(
\begin{array}
[c]{cc}%
1+\mathbb{H}_{+}(x,t) & i\left(  \mathbb{H}\left(  \Phi\right)  \right)
^{1/2}\\
i\left(  \mathbb{H}\left(  \Phi\right)  \right)  ^{1/2} & 1
\end{array}
\right) \\
&  =-2\partial_{x}^{2}\log\det\left(
\begin{array}
[c]{cc}%
1+\mathbb{H}_{+}(x,t) & -\mathbb{H}\left(  \Phi\right) \\
1 & 1
\end{array}
\right) \\
&  =-2\partial_{x}^{2}\log\det\left(
\begin{array}
[c]{cc}%
1+\mathbb{H}_{+}(x,t) & 1\\
-\mathbb{H}\left(  \Phi\right)  & 1
\end{array}
\right)  ,
\end{align*}
where $u_{+}(x,t)$ is the solution to (\ref{KdV}) with data $q_{+}$ and
$\mathbb{H}_{+}(x,t)=\mathbb{H}\left(  \varphi_{+}\right)  $.
\end{theorem}

This theorem is a manifestation of the unidirectional nature of the KdV
equation. The effect of the part of initial data supported on $\left(
-\infty,0\right)  $ is encoded in the Hankel operator $\mathbb{H}\left(
\Phi\right)  $ with an analytic symbol, while the part $\mathbb{H}_{+}(x,t)$
is solely determined by the data on $\left(  0,\infty\right)  $. Theorem
\ref{sep of inft for kdv} provides a convenient starting point to extending
the IST\ formalism to initial data $q$ beyond the realm of the short range
scattering. Since, in general, there is no inverse scattering procedure
available outside of the short range setting we have to rely on suitable
limiting arguments.

\section{Proof of the Main Theorem\label{main}}

With most of ingredients prepared in the previous sections very little is left
to prove Theorem \ref{MainThm}. Take $b<0$ and consider the problem
(\ref{KdV}) with initial data $q_{b}=\left.  q\right\vert _{\left(
b,\infty\right)  }$. By Theorem \ref{sep of inft for kdv} for its solution we
have
\begin{equation}
u_{b}(x,t)=-2\partial_{x}^{2}\log\det\left\{  1+\mathbb{H}_{+}(x,t)+\mathbb{H}%
\left(  \Phi_{b}\right)  \right\}  ,\label{u_b}%
\end{equation}
where%
\begin{align*}
\Phi_{b}\left(  k\right)   &  =-\frac{k+i}{2\pi i}\int_{\mathbb{R}+ih}%
\frac{\xi_{x,t}(\lambda)G_{b}(\lambda)}{\left(  \lambda+i\right)  \left(
\lambda-k\right)  }\ d\lambda,\ \ \ h>\max\left(  \kappa_{n}^{b}\right)  ,\\
G_{b} &  :=\frac{T_{+}^{2}R_{b}}{1-L_{+}R_{b}},
\end{align*}
and is the right reflection coefficient from $\left.  q\right\vert _{\left(
b,0\right)  }$. As is well-known (see, e.g. \cite{Deift79}), $R_{b}$ is a
meromorphic function on the entire plane, and \cite{GruRybSIMA15} uniformly on
compacts in $\mathbb{C}^{+}$
\begin{equation}
R_{b}\left(  \lambda\right)  \rightarrow\frac{i\lambda-m_{-}\left(
\lambda^{2}\right)  }{i\lambda+m_{-}\left(  \lambda^{2}\right)  }:=R\left(
\lambda\right)  ,\ \ b\rightarrow-\infty,\label{eq8.2}%
\end{equation}
where $m_{-}\left(  k^{2}\right)  $ is the Titchmarsh-Weyl m-function of
$\mathbb{L}_{q_{-}}^{D}$, the Schr\"{o}dinger operator on $L^{2}\left(
-\infty,0\right)  $ with a Dirichlet boundary condition at $0$. As is
well-known, $m_{-}\left(  \lambda\right)  $ is analytic on $\mathbb{C}$ away
from the spectrum of $\mathbb{L}_{q_{-}}^{D}$ which due to the condition
\ref{cond 1} is bounded from below (see e.g. \cite{Eastham74}). Consequently,
$R$ is\footnote{$R$ can be interpreted as the (right) reflection coefficient
from $q_{-}$ (see \cite{GruRybSIMA15}, \cite{RybPAMS18} for details).}
analytic in $\mathbb{C}^{+}$ away from purely imaginary points $\lambda$ such
that $\lambda^{2}$ is in the negative spectrum of $\mathbb{L}_{q_{-}}^{D}$.
Thus
\begin{equation}
\lim_{b\rightarrow-\infty}G_{b}=\frac{T_{+}^{2}R}{1-L_{+}R}=:G\label{limG}%
\end{equation}
is an analytic function on $\mathbb{C}^{+}$ away from a bounded set on the
imaginary line. In turn this means that $\Phi=\lim_{b\rightarrow-\infty}%
\Phi_{b}$ is an entire function and $\mathbb{H}\left(  \Phi_{b}\right)
\rightarrow\mathbb{H}\left(  \Phi\right)  \mathbb{\ }$in trace norm. Following
same arguments as in Section \ref{HO} (see also \cite{RybPAMS18} for more
details) we see that for every $n,m$%
\begin{equation}
\left\Vert \partial_{x}^{n}\partial_{t}^{m}\left[  \mathbb{H}\left(  \Phi
_{b}\right)  -\mathbb{H}\left(  \Phi\right)  \right]  \right\Vert
_{\mathfrak{S}_{1}}\rightarrow0,\ \ \ b\rightarrow-\infty.\label{1}%
\end{equation}
Turn now to
\[
\mathbb{H}_{+}(x,t)=\mathbb{H}\left(  \varphi_{+}\right)  =\mathbb{H}\left(
\phi_{+}\right)  +\mathbb{H}\left(  \xi_{x,t}R_{+}\right)  .
\]%
\[
\varphi_{x}\left(  k\right)  =\sum_{n}\frac{c_{n}e^{-2\kappa_{n}x}%
\,}{ik+\kappa_{n}}+e^{2ikx}R(k).
\]
Since $\phi_{+}$ is a rational function, $\mathbb{H}\left(  \phi_{+}\right)  $
is smooth in $\left(  x,t\right)  $ in trace norm.

By (\ref{rep for R+}) we have%
\begin{equation}
\mathbb{H}\left(  \xi_{x,t}R_{+}\right)  =\mathbb{H}\left(  \varphi
_{0}\right)  +\mathbb{H}\left(  \varphi_{1}\right)  , \label{phi 0 and 1}%
\end{equation}
where%
\begin{align*}
\varphi_{0}\left(  k\right)   &  :=\frac{T_{+}\left(  k\right)  }{2ik}%
\xi_{x,t}\left(  k\right)  \int_{0}^{\infty}e^{-2ikx}q\left(  x\right)  dx,\\
\varphi_{1}\left(  k\right)   &  :=\frac{T_{+}\left(  k\right)  }{\left(
2ik\right)  ^{2}}\xi_{x,t}\left(  k\right)  \int_{0}^{\infty}e^{-2ikx}%
Q^{\prime}\left(  x\right)  dx.\ \
\end{align*}
We remind that $\varphi_{0},\varphi_{1}$ are both bounded at $k=0$ as
$T_{+}\left(  k\right)  $ vanishes at $k=0$ to order 1\footnote{In fact, it
happens generically. For the so-called exceptional potentials $T\left(
0\right)  \neq0$ but an arbitrarily small perturbation turns such a potential
into generic. In our case it can be achieved by merely shifting the data $q$
(the KdV is translation invariant).}. Apparently,
\begin{equation}
\partial_{x}\mathbb{H}\left(  \varphi_{0}\right)  =\mathbb{H}\left(  T_{+}%
\xi_{x,t}G_{0}\right)  ,\ \ \ \partial_{x}^{2}\mathbb{H}\left(  \varphi
_{1}\right)  =\mathbb{H}\left(  T_{+}\xi_{x,t}G_{1}\right)  . \label{deriv}%
\end{equation}
where%
\[
G_{0}\left(  k\right)  :=\int_{0}^{\infty}e^{-2ikx}q\left(  x\right)
dx,\ \ G_{1}\left(  k\right)  :=\int_{0}^{\infty}e^{-2ikx}Q^{\prime}\left(
x\right)  dx.
\]
The symbols $f\xi_{x,t}G_{0,1}$ in (\ref{deriv}) are different from the ones
studied in Section \ref{HO} by a factor $T_{+}$ of the form $T_{+}=h/B$, where
$h\in H_{+}^{\infty}$ and $B$ is the finite Blaschke product with simple zeros
at $\left(  i\kappa_{n}^{+}\right)  $. This is however a purely technical
circumstance in the way of applying Theorem \ref{Trace class theorem}. The
easiest way to circumvent it is to alter our original $q$ by performing the
Darboux transform on $q_{+}$ removing all (negative) bound states of
$\mathbb{L}_{q_{+}}$. Then $T_{+}=h\in H_{+}^{\infty}$. But if $h\in
H_{+}^{\infty}$ and $\varphi\in L^{\infty}$ one easily sees that\footnote{Note
that $h\left(  -k\right)  =\overline{h\left(  k\right)  }$.}
\[
\mathbb{H}(h\varphi)=\mathbb{T}\left(  \overline{h}\right)  \mathbb{H}%
(\varphi),
\]
where $\mathbb{T}\left(  \overline{h}\right)  =\mathbb{P}_{+}\overline{h}$ is
the Toeplitz operator with symbol $\overline{h}$. The letter is a bounded
operator independent of $\left(  x,t\right)  $ and smoothness in trace norm of
$\mathbb{H}\left(  T_{+}\xi_{x,t}G_{0,1}\right)  $ with respect of $\left(
x,t\right)  $ is the same as $\mathbb{H}\left(  \xi_{x,t}G_{0,1}\right)  $. As
is well-known, adding back the previously removed bound states results in
adding solitons corresponding to $-\left(  \kappa_{n}^{+}\right)  ^{2}$ (which
are of Schwartz class).

Recalling from Proposition \ref{on R} that%
\[
\left\vert Q^{\prime}\left(  x\right)  \right\vert \leq C_{1}\left\vert
q\left(  x\right)  \right\vert +C_{2}\int_{x}^{\infty}\left\vert q\right\vert
,
\]
one concludes that if $q\in L_{N}^{1}$ then $Q^{\prime}\in L_{N-1}^{1}$. By
Theorem \ref{Trace class theorem} if $N=5/2$ then $\mathbb{H}\left(  \xi
_{x,t}G_{0}\right)  $ and $\mathbb{H}\left(  \xi_{x,t}G_{1}\right)  $ are
differentiable in $x$ in $\mathfrak{S}_{1}$ four and three times respectively.
By (\ref{phi 0 and 1}) and (\ref{deriv}) $\mathbb{H}\left(  \xi_{x,t}%
R_{+}\right)  $ is differentiable in $x$ in $\mathfrak{S}_{1}$ five times and
hence so is $\mathbb{H}_{+}(x,t)$. Thus, since $\partial_{t}\xi_{x,t}%
=\partial_{x}^{3}\xi_{x,t}$, the formula (\ref{u_b}) defines a classical
solution $u_{b}\left(  x,t\right)  $ with initial data $q_{b}$ and it remains
to let $b\rightarrow-\infty$. But it follows from (\ref{1}) that%
\begin{equation}
\lim_{b\rightarrow-\infty}u_{b}\left(  x,t\right)  =-2\partial_{x}^{2}\log
\det\left\{  1+\mathbb{H}_{+}(x,t)+\mathbb{H}\left(  \Phi\right)  \right\}
:=u(x,t)\label{u}%
\end{equation}
and $u(x,t)$ is a classical solution to (\ref{KdV}). Theorem \ref{MainThm} is proven.

In fact, we have proven a stronger statement

\begin{theorem}
\label{smoothing}If in Theorem \ref{MainThm} $q\in L_{N}^{1}$ then $u\left(
x,t\right)  $ is continuously differentiable $\left\lfloor 2N\right\rfloor -2$
times in $x$ and $\left\lfloor \left(  2N-2\right)  /3\right\rfloor $ times in
$t$.
\end{theorem}

We conclude this section with yet another solution formula, which can be
viewed as a generalized Dyson formula.

\begin{theorem}
Under conditions of Theorem \ref{MainThm}, the solution to (\ref{KdV}) can be
represented by%
\[
u(x,t)=-2\partial_{x}^{2}\log\det\left(  1+\mathbb{H}(\varphi_{x,t})\right)  ,
\]
with%
\begin{equation}
\varphi_{x,t}(k)=\int_{0}^{h_{0}}\frac{\xi_{x,t}(is)\,}{s+ik}d\rho\left(
s\right)  +\xi_{x,t}(k)R(k),\label{symbol}%
\end{equation}
where $R$ is the right reflection coefficient of $q$ and $d\rho$ is a positive
finite measure.
\end{theorem}

Note that the pair $\left(  R,\rho\right)  $ can be viewed as scattering data
associated with $\mathbb{L}_{q}$ and only (\ref{symbol}) needs proving. It is
proven in our \cite{GruRybSIMA15} where a complete treatment of $\left(
R,\rho\right)  $ is also given.

\section{Conclusions\label{conclusions}}

Theorem \ref{smoothing} says that, loosely speaking, the KdV flow\emph{
}instantaneously smoothens any (integrable) singularities of $q\left(
x\right)  $ as long $q\left(  x\right)  =o\left(  x^{-2}\right)  $,
$x\rightarrow+\infty$. Such an effect is commonly referred to as dispersive
smoothing. This smoothing property becomes stronger as the rate of decay at
$+\infty$ increases, the behavior at $-\infty$ playing no role. In
\cite{RybCommPDEs2013} we show that if ($C,\delta>0$)%
\begin{equation}
q\left(  x\right)  =O\left(  \exp\left(  -Cx^{\delta}\right)  \right)
,\ \ \ x\rightarrow+\infty, \label{strong decay}%
\end{equation}
then (1) if $\delta>1/2$ then $u(x,t)$ is meromorphic with respect to $x$ on
the whole complex plane (with no real poles) for any $t>0$; (2) if
$\delta=1/2$ then $u(x,t)$ is meromorphic in a strip around the $x-$axis
widening proportionally to $\sqrt{t}$; (3) for $0<\delta<1/2$ the solution
need not be analytic but is at least Gevrey smooth.

Actually, the requirement that $q$ is locally integrable can be lifted. By
employing the arguments from our \cite{GruRemRyb2015} we may easily extend all
our results to include $H^{-1}$ type singularities (like Dirac $\delta
-$functions, Coulomb potentials, etc.) on any interval $\left(  -\infty
,a\right)  $.

The condition (\ref{cond 1}) is optimal. Indeed, what we actually need is
semiboundedness of $\mathbb{L}_{q}$ from below, which is guaranteed by
(\ref{cond 1}). If $q$ is negative then (\ref{cond 1}) becomes also necessary
\cite{Eastham74}.

The absence of decay at $-\infty$ ruins any hope that classical conservations
laws would take place. We do not however rule out existence of some
regularized conservation laws or at least some energy estimates. It would of
course be important to find such estimates.

\end{document}